\documentclass{article}
\usepackage{PRIMEarxiv}
\usepackage[utf8]{inputenc}
\usepackage[T1]{fontenc}
\usepackage{microtype}
\usepackage[authoryear,round]{natbib}

\usepackage{amsmath,amsfonts,bm}

\def\1{\bm{1}}

\DeclareMathAlphabet{\mathsfit}{\encodingdefault}{\sfdefault}{m}{sl}
\SetMathAlphabet{\mathsfit}{bold}{\encodingdefault}{\sfdefault}{bx}{n}

\newcommand{\E}{\mathbb{E}}

\usepackage{url}
\usepackage{algorithm}
\usepackage{algpseudocode}
\usepackage{graphicx}
\usepackage{amsthm}
\usepackage{booktabs}
\usepackage{placeins}
\usepackage{xcolor}
\usepackage[colorlinks,linkcolor=blue,citecolor=blue,linktoc=page]{hyperref}
\hypersetup{pdftitle={Marrying Pricing and Advertising with LLMs},pdfauthor={Alessandro Barro, Francesco Bacchiocchi, Francesco Emanuele Stradi, Alberto Marchesi}}

\newtheorem{proposition}{Proposition}[section]

\title{Marrying Pricing and Advertising with LLMs}
\author{%
  \normalfont
  \begin{tabular}{@{}c@{\hspace{0.02\textwidth}}c@{}}
    \makebox[0.47\textwidth]{\textbf{Alessandro Barro}} &
    \makebox[0.47\textwidth]{\textbf{Francesco Bacchiocchi}} \\
    Politecnico di Milano & Politecnico di Milano \\
    \texttt{alessandro1.barro@mail.polimi.it} &
    \texttt{francesco.bacchiocchi@polimi.it} \\[1.5em]
    \textbf{Francesco Emanuele Stradi} & \textbf{Alberto Marchesi} \\
    Politecnico di Milano & Politecnico di Milano \\
    \texttt{francescoemanuele.stradi@polimi.it} &
    \texttt{alberto.marchesi@polimi.it}
  \end{tabular}%
}
\date{}

\begin{document}

\maketitle

\begin{abstract}
We study a \emph{sequential pricing} problem in which a seller jointly
posts a \emph{price} and an \emph{advertisement} generated by a
\emph{large language model} (LLM). The seller aims to maximize revenue
under an  unknown product demand that depends on both decisions, while observing only
whether each offer leads to a purchase. We propose an online actor-critic
algorithm that combines \emph{low-rank adaptation} (LoRA) of a pretrained LLM
with a demand model fitted to available data. At each round, the actor
generates an advertisement, and the critic estimates purchase probabilities
to guide price selection. Then, the resulting feedback is used to
update both the actor and the critic, with the critic's revenue estimates
providing a baseline for policy gradient updates of the actor. To evaluate
our approach, we develop an evaluation framework with three synthetic demand models
and a demand simulator built from real-world marketplace data. 
Finally, we compare our algorithm with benchmarks that do not jointly optimize price selection and advertisement generation, achieving expected revenue gains over the reference policy of \textbf{5.69\%}, \textbf{5.18\%} and \textbf{55.96\%} under the three synthetic demand models and \textbf{5.81\%} under the marketplace simulator.
\end{abstract}

\section{Introduction}

Online marketplaces have deeply changed how sellers set prices. Indeed, a seller can revise the price of a product at negligible cost and observe how buyers respond to the new price shortly after. Moreover, a single seller may offer a large catalog of products, so that prices must be set by automated rules rather than by hand \citep{misra2019dynamic}. However, the demand for each product is typically unknown in advance. Thus, the seller must learn which prices are profitable while collecting revenue from the same interactions used for learning. This trade-off is captured by the framework of sequential posted pricing, in which a seller repeatedly offers a product to arriving buyers and uses the resulting purchase feedback to guide subsequent price selection \citep{kleinberg2003value}. Building on this framework, a large body of work studied how to learn an unknown demand from purchase outcomes (see, e.g., \citealp{besbes2009dynamic}, and the survey of \citealp{denboer2015dynamic}).

Notice that the price is not the only decision that shapes demand. Indeed, the seller also decides how to present the product, and the advertisement shown to buyers may change their willingness to purchase. In a large-scale experiment, \citet{bertrand2010advertising} studied loan offers mailed by a lender to more than 50{,}000 former clients, randomizing independently the interest rate and several features of the letter's content. Some content features affected loan demand as much as large price changes. For example, showing a single example loan instead of four had the same effect as lowering the monthly interest rate by $2$ percentage points. Hence, pricing and advertising cannot be optimized separately. Indeed, an advertisement that increases demand also shifts the revenue-maximizing price, and conversely the value of an advertisement depends on the price at which the product is offered.

At the same time, the way sellers produce advertisements is changing. Indeed, product descriptions in online marketplaces are increasingly generated by large language models (LLMs) rather than written entirely by hand. For instance, Amazon released LLM-based tools that generate product titles
and descriptions from a short input provided by the seller.\footnote{www.aboutamazon.com/news/small-business/amazon-sellers-generative-ai-tool} Moreover, recent works fine-tuned LLMs on performance feedback to produce more effective advertisements. \citet{jiang2025improving} trained an LLM through reinforcement learning, using historical advertisement performance as a reward. In a ten-week experiment involving nearly 35{,}000 advertisers, the trained model increased click-through rates by 6.7\% with respect to a supervised imitation model. Similarly, \citet{chen2025ctr} and \citet{wang2026relate} aligned advertisement generation with click-through and conversion rates observed online. However, these works only optimized the advertisement, and did not consider the price of the advertised product.

When advertisements take the form of natural language, a joint pricing and advertising policy must also decide how content is generated and adapted in response to feedback. LLMs offer a natural way to parametrize this decision, generating advertisements conditioned on information about the product and its environment. In this paper, we study a sequential setting where a seller presents each arriving buyer with a price and an LLM-generated advertisement, and observes only whether a purchase occurs. This binary feedback reflects the combined effect of both decisions and reveals nothing about how the buyer would have responded to alternatives. Since the space of advertisements is too large to explore exhaustively, the seller must generalize from previous interactions, both when generating text and when selecting prices. In addition, optimizing solely for revenue may lead the policy to deviate substantially from the reference model, potentially compromising the fluency and coherence of the generated content. Controlling this deviation helps preserve the language capabilities of the pretrained model while allowing the seller to adapt its advertisements. In this work, we address these challenges through an online actor-critic approach that combines adaptation of a pretrained LLM with demand estimation, using purchase outcomes to update both components and guide their joint optimization.

\subsection{Original contribution}
We propose an online actor-critic algorithm for joint price and advertisement optimization, operating in batches of equal size indexed by $\tau$. The actor is a pretrained LLM fine-tuned through low-rank adaptation (LoRA) to generate advertisements, while the critic is a demand model fitted to available data that estimates purchase probabilities given price, advertisement and product-level context. For each advertisement generated in batch $\tau$, the seller selects a price through an exploration rule guided by the critic. With probability $\varepsilon_\tau$, the price is sampled from a given exploration distribution $\nu_\tau$ on $[0,1]$. Otherwise, the seller selects a price maximizing the critic's estimate of expected revenue. The observed purchase outcomes are used to update the critic's demand estimates and the actor's generative policy for the advertisement. The actor updates its LoRA parameters through policy gradient steps to maximize an objective that balances expected revenue with a KL penalty for deviations from the fixed pretrained reference model. These updates use a leave-one-out baseline that averages the critic's predicted revenues instead of observed rewards, reducing the noise introduced by individual purchase outcomes.

Section~\ref{sec:experiments} evaluates expected revenue under explicit synthetic demand and studies transfer to marketplace products unseen during actor training. We consider \emph{logit}, \emph{probit}, and \emph{price disclosure} demand, together with a simulator fitted to marketplace listings. The synthetic models examine how advertisement content affects purchase propensity and price sensitivity. We compare the \texttt{Trained} and \texttt{Reference} policies using the same final critic, isolating advertisement adaptation under shared learned pricing. In the three synthetic environments, expected revenue increases by \textbf{5.69\%}, \textbf{5.18\%}, and \textbf{55.96\%}, respectively. The largest gain is fully explained by the increased frequency of the rewarded content feature under the declared demand mechanism. Under the marketplace simulator obtained from real data, average normalized expected revenue increases by \textbf{5.81\%}, with improvements on 91 of 100 products across five categories. We also compare revenue during learning with \texttt{UCB} under equal online interaction budgets. Appendices~\ref{app:exploration} and~\ref{app:baseline} examine performance with progressively less offline data and the mean squared error of the critic-based baseline, respectively.

\subsection{Related Work}
Sequential pricing with unknown demand has been extensively studied in online learning and revenue management. \citet{kleinberg2003value} studied auctions with posted prices in which a seller learns from purchase decisions and characterized the revenue loss associated with not knowing the demand curve. \citet{besbes2009dynamic} developed pricing policies that learn demand while collecting revenue, considering both parametric and nonparametric demand models. \citet{misra2019dynamic} subsequently proposed a dynamic pricing approach based on multiarmed bandit experiments, incorporating restrictions from economic choice theory to guide price exploration. We refer to \citet{denboer2015dynamic} for a survey of dynamic pricing and learning. The interaction between pricing and advertising also has a long history in economics: \citet{dorfman1954optimal} analyzed the joint choice of price, advertising expenditure and product quality. Focusing on the content of advertisements, \citet{bertrand2010advertising} conducted a field experiment that independently randomized loan prices and advertising features, showing that content can have substantial effects on demand relative to price. More recently, \citet{agrawal2023persuasion} studied joint pricing and advertising in a sequential setting with unknown demand. In their formulation, the seller learns from past purchase responses to select both a price and a signaling scheme that discloses information about product quality.

A growing body of work adapts LLMs for advertisement generation using performance feedback. For instance, \citet{jiang2025improving} trained an advertisement generator with reinforcement learning guided by a reward model fitted to historical advertisement performance and evaluated the resulting model in a large advertising experiment. Focusing on feedback collected online, \citet{chen2025ctr} generated diverse candidate advertisements from effective exemplars, tested them online and optimized the generator using preferences derived from observed click rates. More recently, \citet{wang2026relate} optimized advertisement generation with reinforcement learning, using an objective that combines predicted conversion rates with quality and diversity criteria. These methods improve advertising effectiveness without jointly optimizing the selling price. Alongside these applications, our approach draws on a broader methodological literature on the adaptation of language models. \citet{ziegler2019finetuning} studied learning from human preferences with a KL penalty to control deviations from a reference policy, while \citet{hu2022lora} introduced LoRA, which learns updates of low rank while keeping pretrained weights fixed. For policy optimization, REINFORCE provides a gradient estimator based on sampled outputs and their rewards \citep{williams1992simple}. Building on this estimator, \citet{kool2019buy} studied how multiple samples for the same input can be used to construct baselines that reduce its variance. When it comes to LLMs, \citet{ahmadian2024back} further showed that REINFORCE variants with leave-one-out baselines can effectively support learning from human feedback.

\section{Problem formulation}
\label{sec:formulation}
In the standard formulation of sequential posted pricing~\citep{kleinberg2003value}, at each round a seller posts a price
$p \in [0,1]$ for an item, and an arriving buyer purchases it if their private valuation exceeds $p$.
When buyers' valuations are independently drawn from a fixed distribution, this interaction can be
equivalently described by a demand curve $d : [0,1] \to [0,1]$, where $d(p)$ denotes the probability of
purchase at price $p$. In the real world, however, the seller must specify not only a price $p \in [0,1]$,
but also an advertisement $x \in \mathcal{X}$, where $\mathcal{X}$ denotes the space of possible
advertisements shown to buyers. Each offer is associated with a context vector $z \in \mathcal{Z}$
representing characteristics of the product and the environment, which the seller observes before
choosing the advertisement and price. Since both the advertisement and the context may affect
engagement and the likelihood of purchase, we employ a demand function
$d : [0,1] \times \mathcal{X} \times \mathcal{Z} \to [0,1]$ to model buyer response, where $d(p,x,z)$ is the probability of
purchase at price $p$, given advertisement $x$ and context $z$.

We focus on textual advertisements $x\in\mathcal X$ generated by a large language model $\pi_\theta$ with parametrization $\theta\in\Theta$. For an observed context $z$, we generate $x\sim\pi_\theta(\cdot\mid c(z))$, where $c(z)$ is a context-dependent prompt.\footnote{We assume that the mapping $c(z)$ is fixed and known to the seller.} The seller also selects a price through a pricing rule $q:\mathcal X\times\mathcal Z\to[0,1]$, and $q(x,z)$ denotes the price for advertisement $x$ and context $z$, before observing purchase feedback. In this paper, we restrict $\Theta$ to the parametrizations obtained through low-rank adaptation (LoRA)~\citep{hu2022lora} of a fixed pretrained reference model. We group its parameters as $\theta_{\mathrm{ref}}=((W_{\ell,\mathrm{ref}})_{\ell=1}^{L},\theta_{\mathrm{ref,rest}})$, where $W_{\ell,\mathrm{ref}}\in\mathbb R^{d_\ell\times k_\ell}$ are the weight matrices selected for adaptation and $\theta_{\mathrm{ref,rest}}$ contains the remaining parameters. For each selected matrix, we fix a rank $r_\ell\ll\min(d_\ell,k_\ell)$ and introduce trainable adapter matrices $A_\ell\in\mathbb R^{r_\ell\times k_\ell}$ and $B_\ell\in\mathbb R^{d_\ell\times r_\ell}$. We collect these matrices in $\psi=((A_\ell,B_\ell))_{\ell=1}^{L}\in\Psi$, where $\Psi$ denotes the space of adapter parameters. Thereby, we introduce the \emph{low-rank adaptation map} $\Delta\theta$, which assembles the scaled matrix products $B_\ell A_\ell$ into a parameter update, with zero entries for all unadapted parameters. Formally,
\begin{equation}
\Delta\theta:\psi\longrightarrow
\left(
\left(\frac{\alpha_\ell}{r_\ell}B_\ell A_\ell\right)_{\ell=1}^{L},
0
\right),
\label{eq:low-rank-adaptation-map}
\end{equation}
where $\alpha_\ell>0$ are fixed scaling coefficients. Thus, each selected weight matrix receives an additive update of rank at most $r_\ell$, while the zero block leaves all remaining parameters unchanged. Accordingly, we define
\(
\Theta:=\left\{\theta_{\mathrm{ref}}+\Delta\theta(\psi)\mid\psi\in\Psi\right\}.
\)
Working within $\Theta$ therefore amounts to updating only the adapter parameters $\psi$, while keeping $\theta_{\mathrm{ref}}$ fixed. 

Finally, we let $\mu$ denote the probability distribution of product and environmental contexts on $\mathcal{Z}$ across offers and consider a repeated interaction over $T$ rounds between the seller and a sequence of buyers.\footnote{For $n \in \mathbb{N}_{>0}$, let $[n] := \{1,\ldots,n\}$.} For each round $t\, \in\, [T]$, the interaction with the current buyer is summarized as follows:
\begin{enumerate}
    \item The seller observes the context $z_t \sim \mu$ of an offer to an arriving buyer.
    \item The seller selects a model parametrization $\theta_t\in\Theta$ and a pricing rule
    $q_t$. The seller generates an advertisement $x_t\sim\pi_{\theta_t}(\cdot\mid c(z_t))$,
    sets the price $p_t=q_t(x_t,z_t)$, and displays the advertisement-price pair $(x_t,p_t)$.\footnote{ The seller may randomize the choice of $q_t$ to explore prices.}
    \item The seller observes the binary purchase outcome $y_t\in\{0,1\}$, where $y_t=1$ if the buyer
    purchases the item and $y_t=0$ otherwise, with
    $y_t\mid p_t,x_t,z_t\sim\mathrm{Ber}\bigl(d(p_t,x_t,z_t)\bigr)$.
    \item The seller receives revenue $r_t:=p_ty_t$.\footnote{We use the terms reward and revenue interchangeably.}
\end{enumerate}
This model implies that the conditional expected selling probability and revenue satisfy, respectively,
\[
\mathbb{E}[y_t \mid p_t,x_t,z_t] = d(p_t,x_t,z_t),
\qquad
\mathbb{E}[r_t \mid p_t,x_t,z_t] = p_t\,d(p_t,x_t,z_t),
\]
For a pricing rule $q : \mathcal{X} \times \mathcal{Z}\! \to\! [0,1]$, model parametrization $\theta \!\in\! \Theta$ and fixed $z$, expected revenue is:
\[
u(q,\theta \mid z ) :=
\mathbb{E}_{x \sim \pi_\theta(\cdot \mid c(z))}
\bigl[q(x,z)\,d(q(x,z),x,z)\bigr].
\]
If the seller's sole goal were to maximize expected revenue, the objective would be to learn a pair $(q,\theta)$ maximizing $u(q,\theta \mid z)$ for every context $z \in \mathcal{Z}$. In practice, however, optimizing only for revenue may drive
the distribution of generated advertisements too far from the reference language policy. To mitigate this drift, we consider a KL-regularized objective with respect to a fixed reference model
$\pi_{\theta_\mathrm{ref}}$. This approach is analogous to KL-regularized reinforcement learning from
human feedback (RLHF), where a policy is optimized for reward while being penalized for deviating
from a reference model~\citep{ziegler2019finetuning}. Thus, for any pricing rule $q$ and
parametrization $\theta\in\Theta$, we define:
\begin{equation}\label{eq:regularized-objective}
u_\beta(q,\theta \mid z ) := u(q,\theta \mid z) - \beta\,
D_{\mathrm{KL}}\!\left(\pi_\theta(\cdot \mid c(z)) \,\|\, \pi_{\theta_\mathrm{ref}}(\cdot \mid c(z))\right),
\end{equation}

where $\beta>0$ controls the trade-off between expected revenue and proximity to the reference model.\footnote{ For two probability distributions $P$ and $Q$ on a discrete space $\mathcal S$, the Kullback-Leibler divergence is defined as \( D_{\mathrm{KL}}(P\|Q)
		:=\sum_{s\in\mathcal S}P(s)\log( {P(s)}/{Q(s)} ), \) where terms with $P(s)=0$ contribute zero. }
This discourages overoptimization of the objective by penalizing deviations from the original policy.

\section{Online actor-critic algorithm}
\label{sec:algorithm}

In this section, we introduce Algorithm~\ref{alg:actor-critic}, which uses an actor-critic approach to jointly learn price selection and advertisement generation. In our approach, the actor is a pretrained LLM adapted through low-rank adaptation (LoRA) to generate advertisements, while the critic estimates purchase probabilities given a price, an advertisement, and the context. We describe our algorithm for a fixed context $z$. We first divide the time horizon $T$ into $N=T/B$ batches of size $B$, assuming, for ease of presentation, that $T$ is a multiple of $B$ and that $B\geq 2$. Before the first batch, we initialize the LoRA adapter so that the actor coincides with the pretrained reference model, whose parameters $\theta_{\mathrm{ref}}$ remain fixed throughout the algorithm. Then, at the beginning of each batch $\tau\in[N]$, we set $\theta_\tau=\theta_{\mathrm{ref}}+\Delta\theta(\psi_\tau)$ using the current adapter parameters $\psi_\tau$, as specified at Line~\ref{line:jpa-policy}. During each batch, the actor generates advertisements, while the critic, parametrized by $\phi_\tau$, estimates purchase probabilities given a price, an advertisement, and the context. We keep the actor and critic parameters fixed throughout each batch.
\begin{algorithm}[t]
    {
    \setlength{\abovedisplayskip}{2pt}
    \setlength{\belowdisplayskip}{2pt}
    \setlength{\abovedisplayshortskip}{1pt}
    \setlength{\belowdisplayshortskip}{2pt}
	\caption{Pricing \& Advertising via Low-Rank Adaptation and REINFORCE Leave-One-Out}\label{alg:actor-critic}
	\begin{algorithmic}[1]
		\Require Fixed context $z$, actor reference model $\theta_{\mathrm{ref}}$, low-rank adapter $\psi_1$, critic reference model $\phi_1$, actor and critic batch sizes $B,S\geq 2$, integer critic:actor update ratio $K\geq1$, time horizon $T$, KL coefficient $\beta$, critic regularization coefficient $\rho$, exploration distributions $\nu_\tau$ on $[0,1]$, exploration probabilities $\varepsilon_\tau$, actor learning rate $\eta_\tau$, critic learning rate $\gamma_\tau$, for every $\tau\!\in\![T/B]$
		
		\State Set $N=T/B$ and initialize the buffer $\mathcal B=\emptyset$
		\For{$\tau=1,\ldots,N$}
		\State Set $\theta_\tau=\theta_{\mathrm{ref}}+\Delta\theta(\psi_\tau)$ \label{line:jpa-policy}
		\For{$b=1,\ldots,B$}\label{line:jpa-actor-condition}
		\State Sample an advertisement $x_{\tau,b}\sim\pi_{\theta_\tau}(\,\cdot\,\mid\,c(z))$\label{line:jpa-ad}
		\State With probability $\varepsilon_\tau$, sample a price $p_{\tau,b}\sim \nu_\tau$; otherwise select it according to \label{line:jpa-price}
		\[
		p_{\tau,b} \in
			\displaystyle\arg\max_{p\in[0,1]} p\,\hat d_{\phi_\tau}(p,x_{\tau,b},z)
		\]
		\State Observe the feedback $y_{\tau,b}\sim\operatorname{Ber}(d(p_{\tau,b},x_{\tau,b}, z))$ and the reward $r_{\tau,b}=p_{\tau,b}\,y_{\tau,b}$  \label{line:jpa-feedback}
		\EndFor
		\If{$\tau\mod K=0$}\label{line:jpa-actor-condition-2}
		\State For every $j,b\in[B]$, compute the revenue advantage estimates\label{line:jpa-advantage}
        {
        \setlength{\belowdisplayskip}{0pt}
		\[
		\widehat u_{\phi_\tau,j}:=p_{\tau,j}\widehat d_{\phi_\tau}(p_{\tau,j},x_{\tau,j},z),
		\qquad
		A_{\tau,b}:=r_{\tau,b}-\frac{1}{B-1}\sum_{j\in[B]:\,j\neq b}\widehat u_{\phi_\tau,j}
		\]
        }
		\State Define the actor surrogate loss\label{line:jpa-actor-loss}
		\[
		\mathcal L_{\mathrm{actor}}(\psi)
		=-\frac1B\sum_{b=1}^{B}
		\left(A_{\tau,b}-\beta\log
		\frac{\pi_{\theta_\tau}(x_{\tau,b}\mid c(z))}
		{\pi_{\theta_{\mathrm{ref}}}(x_{\tau,b}\mid c(z))}\right)
		\log\pi_{\theta_{\mathrm{ref}}+\Delta\theta(\psi)}(x_{\tau,b}\mid c(z))
		\]
		\State Update $\psi_\tau$ to $\psi_{\tau+1}$ with one optimization step on $\mathcal L_{\mathrm{actor}}$ using learning rate $\eta_\tau$\label{line:jpa-actor-step}
		\Else
		\State Set $\psi_{\tau+1}=\psi_\tau$
		\EndIf
		\State For every $b \in [B]$, let $e_{\tau,b}=(x_{\tau,b},p_{\tau,b},y_{\tau,b})$ and update the buffer $\mathcal B = \mathcal B \cup(e_{\tau,b})_{b=1}^B$\label{line:jpa-buffer}
		\State Sample $(e_s)_{s=1}^S$ uniformly from the buffer $\mathcal B$ and define the critic loss
		\[
		\mathcal L_{\mathrm{critic}}(\phi)
		=-\frac{1}{S}\sum_{s=1}^S \left[ y_s\log \hat d_{\phi}(p_s,x_s, z) +(1-y_s)\log\left(1-\hat d_{\phi}(p_s,x_s, z)\right)\right]+\rho\,\lVert\phi-\phi_1\rVert_2^2
		\]
		\State Update $\phi_\tau$ to $\phi_{\tau+1}$ with one optimization step on $\mathcal L_{\mathrm{critic}}$ using learning rate $\gamma_\tau$\label{line:jpa-critic-step}
		\EndFor
		
		\State \Return $\psi_{N+1},\phi_{N+1}$
	\end{algorithmic}
    }
\end{algorithm}

Given a batch $\tau$, we collect $B$ seller-buyer interactions indexed by $b=1,\ldots,B$. In each interaction $b\in[B]$, the seller first generates an advertisement $x_{\tau,b}\sim\pi_{\theta_\tau}(\cdot\mid c(z))$.\footnote{In Algorithm~\ref{alg:actor-critic}, for $t=(\tau-1)B+b$, the quantities $x_{\tau,b}$, $p_{\tau,b}$, $r_{\tau,b}$ and $y_{\tau,b}$ correspond to $x_t$, $p_t$, $r_t$, and $y_t$, respectively, in the learning interaction. Similarly, $\theta_\tau$ corresponds to $\theta_t$ for all $t$ such that $(\tau-1)B+1\leq t\leq \tau B$. The actor parameters and the pricing strategy remain fixed throughout the batch.} Then, the seller selects a price according to an $\varepsilon_\tau$-greedy pricing rule, as specified at Line~\ref{line:jpa-price}. In particular, with probability $1-\varepsilon_\tau$, the seller selects a price maximizing the expected revenue estimated using the critic, namely,
\[
p_{\tau,b}\in\arg\max_{p\in[0,1]}p\,\widehat d_{\phi_\tau}(p,x_{\tau,b},z).
\]
With probability $\varepsilon_\tau$, the seller instead samples $p_{\tau,b}\sim\nu_\tau$.
We choose $\nu_\tau$ before collecting batch $\tau$ and hold it fixed within that batch. Appendix~\ref{app:exploration} instantiates this choice under progressively missing offline data.
After displaying the advertisement-price pair $(x_{\tau,b},p_{\tau,b})$, the seller observes the purchase outcome $y_{\tau,b}$ and computes the realized revenue $r_{\tau,b}=p_{\tau,b}y_{\tau,b}$ at Line~\ref{line:jpa-feedback}. By repeating these steps for all $b\in[B]$,
	we collect a batch of advertisement-price pairs and their
	corresponding purchase outcomes. These observations provide
	the feedback for the subsequent model updates. Once the $B$ seller-buyer interactions in batch $\tau$ are complete, we update the actor if the condition at Line~\ref{line:jpa-actor-condition-2} is satisfied, \emph{i.e.}, once every $K$ batches. The update is based on REINFORCE~\citep{williams1992simple}. In its basic form, REINFORCE takes an optimization step on the loss:
\begin{equation}\label{eq:reinforce}
	\mathcal L_{\mathrm{RF}}(\psi)
	=-\frac{1}{B}\sum_{b=1}^{B}
	r_{\tau,b}
	\log\pi_{\theta_{\mathrm{ref}}+\Delta\theta(\psi)}
	(x_{\tau,b}\mid c(z)).
\end{equation}
Thus, advertisements that produce higher observed
revenues receive greater weight in the update.
Building on this, we introduce a leave-one-out baseline and a KL penalty. Specifically, we replace the observed revenue $r_{\tau,b}$
in Equation~\eqref{eq:reinforce} with a revenue advantage estimate
$A_{\tau,b}$ minus a KL regularization term and obtain the actor loss defined at Line~\ref{line:jpa-actor-loss}.
To estimate the advantage, we subtract
from $r_{\tau,b}$ a leave-one-out baseline that averages the critic's
expected revenue estimates for $j\neq b$:
\begin{equation}\label{eq:jpa-revenue-advantage}
A_{\tau,b}:=r_{\tau,b}
-\frac{1}{B-1}\sum_{j\in[B]:\,j\neq b}\widehat u_{\phi_\tau,j},
\qquad
\widehat u_{\phi_\tau,j}:=p_{\tau,j}\widehat d_{\phi_\tau}(p_{\tau,j},x_{\tau,j},z).
\end{equation}
Unlike a reward-based REINFORCE leave-one-out baseline, we average the critic's expected revenue estimates $\widehat u_{\phi_\tau,j}$ rather than the observed
	revenues $r_{\tau,j}$ for $j\neq b$. This choice aims to reduce purchase-outcome noise in the baseline. Appendix~\ref{app:baseline} compares this noise with the critic estimation error that replaces it.
In addition, the KL term in the actor loss at Line~\ref{line:jpa-actor-loss}
	plays the role of preventing large deviations from the reference model. As shown in Appendix~\ref{app:unbiased-actor-gradient},
conditional on the history before batch $\tau$, the negative
gradient of this loss evaluated at $\psi_\tau$ is an unbiased
estimator of the gradient with respect to the LoRA parameters
of the KL-regularized objective $u_\beta$ defined in Equation~\eqref{eq:regularized-objective},
averaged over the $\varepsilon_\tau$-greedy pricing strategy
used in that batch. The pricing strategy is held fixed when
taking this gradient. Finally, we store the observations
	collected during the batch in $\mathcal B$ at
	Line~\ref{line:jpa-buffer} and sample $S$ observations to update
	the critic. 
    
    Then, at Line~\ref{line:jpa-critic-step}, we take one gradient step
	on the critic loss, which combines binary cross-entropy on the
	observed purchase outcomes $y_s$ with a penalty for deviations
	from $\phi_1$. In this way, the critic is updated after
	every batch, while the actor is updated every $K$ batches using
	only the seller-buyer interactions in the current batch.

\begin{figure}[H]
	\centering
	\includegraphics[width=0.8\linewidth]{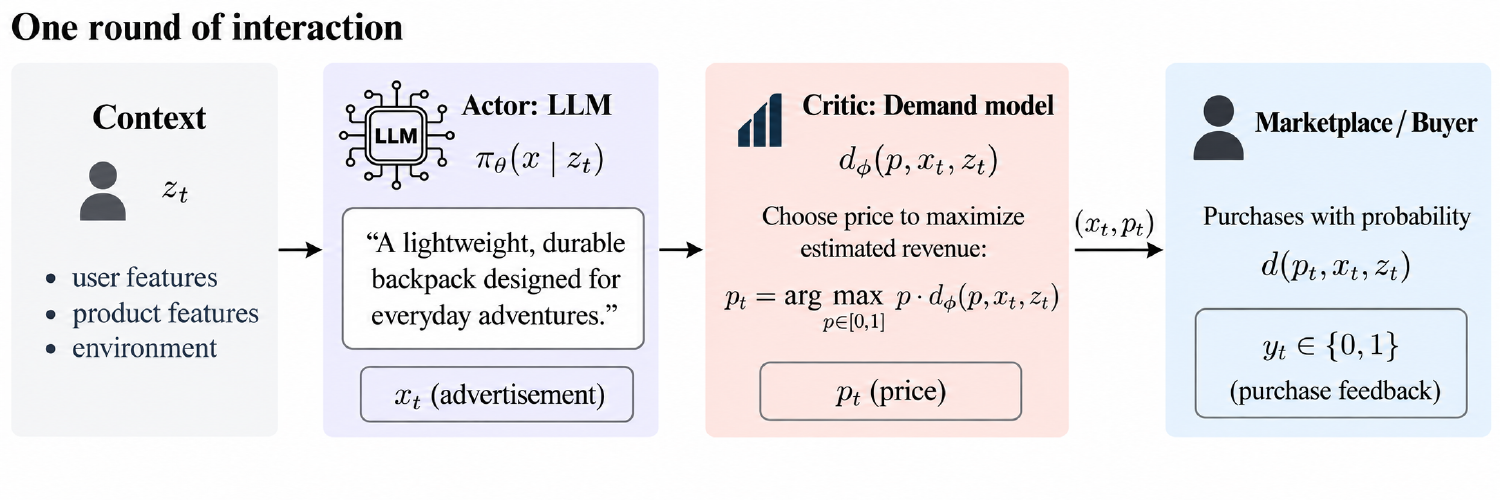}
	\caption{One round of seller-buyer interaction. The actor generates an advertisement, the seller uses the critic to select a price, and the buyer provides purchase feedback. }
	\label{fig:one-round-interaction}
\end{figure}

\providecommand{\experimentroot}{experimental_section}
\section{Experimental Results}
\label{sec:experiments}
\label{exp:body-start}
We evaluate Algorithm~\ref{alg:actor-critic} under three synthetic demand models and a simulator fitted to real-world marketplace data. In the following, we denote by \texttt{Trained} the algorithm that uses the actor and the critic obtained from Algorithm~\ref{alg:actor-critic}, while we denote by \texttt{Reference} the algorithm that selects prices according to the critic model returned by Algorithm~\ref{alg:actor-critic} and generates advertisements according to the reference model, \emph{i.e.}, the model that is not adapted during the seller-buyer interactions. This comparison isolates the effect of adapting the advertisement-generation policy while using the same pricing rule for both approaches. We report both absolute and relative changes in expected revenue, defined as follows:
\[
\Delta u := u_{\texttt{Train}}-u_{\texttt{Ref}},
\qquad
\Delta u(\%) := 100\,\frac{\Delta u}{u_{\texttt{Ref}}},
\]
where $u_{\texttt{Train}}$ and $u_{\texttt{Ref}}$ denote estimates of the expected revenues of the two policies.\\
We also evaluate revenue achieved throughout the learning process. To do so, we compare \texttt{Trained} with \texttt{UCB}. The latter algorithm generates advertisements using the fixed pretrained reference model and selects prices using UCB1~\citep{auer2002finite}, independently of the generated advertisement. Both approaches use the same price grid and interact with buyers over the same number of rounds. At each round $t$, \texttt{UCB} updates its pricing strategy based solely on the observed reward $r_t=p_ty_t$, without relying on a pretrained critic. We present this comparison through training-revenue curves, which show how expected revenue evolves throughout the interactions.
In both the synthetic and real-world experiments, we adapt \texttt{Qwen2.5-1.5B-Instruct} with rank-eight LoRA over $N=500$ batches of $B=32$ offers per run ($T=16{,}000$ purchase outcomes). Appendix~\ref{app:experiments} gives further implementation details, as well as additional ablations.

\subsection{Synthetic Demand Models}
\label{sec:exp-synthetic}
We first examine our performances under \emph{logit} and \emph{probit} demand, two models used to study consumer responses to credit offers~\citep{farrell2021heterogeneity,bertrand2010advertising}. For a fixed context $z$ and a content feature $h(x)$, \textit{e.g.}, an indicator of whether the advertisement mentions a word of choice, these models are given by:
\begin{equation}
d_{\mathrm{logit}}(p,x)=\sigma(a-\kappa p+w^\top h(x)),
\qquad
d_{\mathrm{probit}}(p,x)=\Phi(a-\kappa p+w^\top h(x)).
\label{eq:exp-choice-models}
\end{equation}
Here, $\kappa>0$ ensures that demand decreases with price, while $\sigma$ and $\Phi$ denote the logistic function and standard normal CDF, respectively. For the steel-bottle experiments (see the Appendix for additional details), $h(x)$ equals one if advertisement $x$ contains the word \textit{durable} and zero otherwise. We set illustrative logit coefficients $(a,\kappa,w)=(1,4,0.8)$ and calibrate probit to match the resulting logit purchase probabilities at three fixed price-feature combinations. 

The third model, \emph{price disclosure}, uses $h(x)$ to indicate whether the advertisement reveals the offered price, allowing disclosure to affect both the probability of purchase at a given price and how demand responds to price changes. We construct separate demand curves from the purchase rates reported by \citet{dube2026effects} at \$0.99, \$2.99 and \$5.99 for advertisements that omit the price and those that disclose it.
For each condition, we linearly interpolate between the three observed rates to define purchase probabilities at intermediate prices. This yields two demand curves, with $h(x)$ selecting which curve applies to a generated advertisement. We restrict prices to the observed range and divide them by $5.99$, obtaining the normalized interval $p\in[0.99/5.99,1]$ without changing the associated purchase probabilities.

Across all three models, we pretrain the critic on 8,192 offline offers with \texttt{Reference}-generated advertisements and Bernoulli purchase outcomes sampled from the corresponding demand model. The critic observes the price and the specified feature $h(x)$ without access to the demand parameters; Appendix~\ref{app:synthetic} provides calibration and sampling details.
\begin{table}[h]
\centering\small
\caption{Synthetic expected revenue from 10,000 fresh advertisements per policy and environment. The 95\% intervals use 2,000 bootstrap resamples, drawn independently for the compared policies.}
\label{tab:exp-synthetic}
\vspace{7pt}
\begin{tabular}{@{}lrrr@{}}
\toprule
Policy / metric & Logit & Probit & Price disclosure \\
\midrule
\emph{Trained} & 0.217117 & 0.205616 & 0.027162 \\
\emph{Reference} & 0.205418 & 0.195496 & 0.017416 \\
\midrule
$\boldsymbol{\Delta u}$ & \textbf{+0.011698} & \textbf{+0.010120} & \textbf{+0.009746} \\
$\boldsymbol{\Delta u(\%)}$ & \textbf{+5.69\%} & \textbf{+5.18\%} & \textbf{+55.96\%} \\
95\% interval for $\Delta u$ & [0.0111, 0.0123] & [0.0095, 0.0107] & [0.0096, 0.0099] \\
\bottomrule
\end{tabular}
\end{table}

Table~\ref{tab:exp-synthetic} reports expected revenue gains of \textbf{5.69\%}, \textbf{5.18\%}, and \textbf{55.96\%} over \texttt{Reference} under \emph{logit}, \emph{probit}, and \emph{price disclosure}, respectively. These gains accompany an increase in the frequency of the rewarded content features: the fraction of advertisements containing \textit{durable} rises from $82.14\%$ to $97.11\%$ in logit and $96.73\%$ in probit, while the fraction disclosing the price rises from $35.17\%$ to $98.23\%$. Under the shared critic, the increased disclosure frequency fully accounts for the revenue gain in the \textit{price disclosure} model (see Appendix~\ref{app:synthetic} for further details). The upper panels of Figure~\ref{fig:exp-synthetic} complement these averages with held-out scatter plots, with all block means lying above the equality line across the three demand models. The lower panels display the corresponding training-revenue curves, \textit{i.e.} how expected revenue evolves during the optimization process. As described in the evaluation protocol, we compare \texttt{Trained} and \texttt{UCB} in this online regime. Since the latter continues to balance price exploration and exploitation as feedback accumulates, it makes sense to assess its performance on the fly, rather than evaluating a final price recommendation offline.

\begin{figure}[!htb]
\centering
\includegraphics[width=\linewidth]{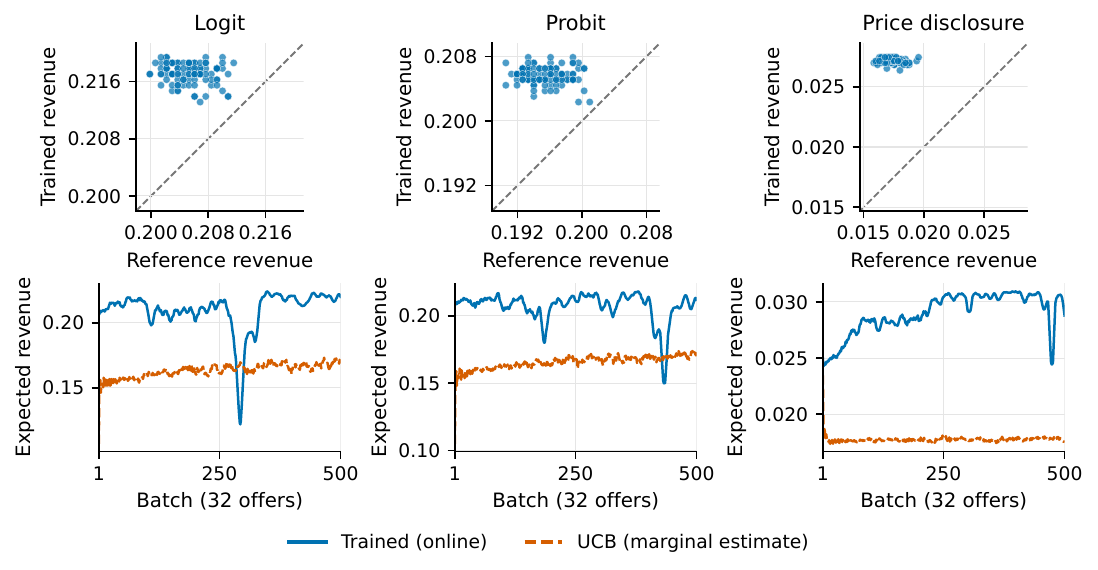}
\caption{Synthetic evaluation. Top: held-out mean revenues (\texttt{Reference}, horizontal; \texttt{Trained}, vertical), with a dashed equality line. Each point averages 100 independently generated advertisements per policy; independent blocks are aligned for display only. Bottom: trailing ten-batch revenue means for \texttt{Trained} (blue) and \texttt{UCB} (orange). \texttt{Trained} revenue uses recorded offer probabilities; \texttt{UCB} revenue is estimated from logged prices and feature frequencies (Appendix~\ref{app:ucb-reconstruction}).}
\label{fig:exp-synthetic}
\end{figure}
\subsection{Experiments Based on Real Data}
\label{sec:exp-real}
The first evaluation protocol assesses the effectiveness of LoRA through a held-out comparison of \texttt{Trained} and \texttt{Reference}. We evaluate expected revenue using a frozen demand simulator fitted to Avito marketplace listings~\citep{avito2018demand}, with text and price as inputs and deal probabilities as targets. Furthermore, we pretrain the \textit{initial} critic on a disjoint split of the same size. The listings cover five product categories: \textit{Audio/Video}, \textit{Home Appliances}, \textit{Laptops}, \textit{Mobile Phones} and \textit{Tablets/E-readers}. Within each category, we run Algorithm~\ref{alg:actor-critic} on a single product to obtain the \texttt{Trained} policy, which we then compare with \texttt{Reference} on 20 other products. Both use the same \textit{final} critic for price selection. We thus generate 100 advertisements and prices per policy and test product, yielding 10,000 evaluations each. During both training and evaluation, five fixed marketplace examples guide the format and phrasing of generated advertisements to limit shifts from the text distribution encoded by the demand. Finally, the context $z$ contains product facts such as name, category and condition. 

We normalize prices separately for each product to use a common action space $[0,1]$ across different monetary price scales. For a product described by context $z$, action $p\in[0,1]$ corresponds to the monetary price $\lambda(z)p$, where $\lambda(z)$ is the product's maximum retained listing price, determined without using demand targets. Both the marketplace simulator and the critic receive this monetary price, matching the kind used while fitting. During evaluation, we divide expected monetary revenue $\lambda(z)p\,d(p,x,z)$ by $\lambda(z)$ and report $p\,d(p,x,z)$. This measures revenue relative to each product's price scale, avoiding greater weight for products simply because their monetary prices are higher. Notably, since $\lambda(z)$ is fixed and positive for each product, this normalization preserves the price that maximizes unregularized revenue within the admissible range. Table~\ref{tab:exp-real} reports an increase in average normalized expected revenue from \textbf{0.1991} for \texttt{Reference} to \textbf{0.2107} for \texttt{Trained}, corresponding to a gain of \textbf{5.81\%}. All five categories improve on average, with gains ranging from \textbf{1.98\%} for \textit{Mobile Phones} to \textbf{14.26\%} for \textit{Laptops}.

The second evaluation protocol examines expected revenue during online learning by comparing \texttt{Trained} with \texttt{UCB} on the five training products. Both methods use the same price grid and online interaction budget, while \texttt{UCB} pairs the reference generator with price exploration guided by observed rewards, without a pretrained critic. Figure~\ref{fig:exp-real} tracks their expected revenue throughout training: \texttt{Trained} achieves higher mean expected revenue over the final 100 batches in every category. This comparison illustrates the benefits of joint price and advertisement optimization supported by offline critic supervision. Appendix~\ref{app:exploration} further examines performance and the role of price exploration as offline data are progressively reduced, while Appendix~\ref{app:baseline} studies the estimation error of the leave-one-out baseline constructed from the critic's revenue estimates (Equation~\eqref{eq:jpa-revenue-advantage}).

\section{Conclusion and Future Work}

We study a sequential pricing problem in which a seller jointly selects a price
and an LLM-generated advertisement, observing only whether the buyer purchases.
We propose an online actor-critic algorithm that adapts the advertisement policy
through LoRA, while a demand model guides price selection and provides a
leave-one-out baseline for the KL-regularized actor updates. Finally, we
validate our algorithm on both synthetic and real-world marketplace data.
A natural extension is to allow demand to change over time. This setting is
challenging because shifting buyer preferences may alter both price sensitivity
and the effect of advertisement content, so the critic must track a moving
target from binary feedback while the actor keeps adapting to its estimates.

\begin{table}[H]
\centering\small
\caption{Marketplace expected revenue from 100 fresh advertisements per policy and product. The 95\% intervals use 5,000 bootstrap resamples, drawn over products and paired across the policies.}
\label{tab:exp-real}
\vspace{7pt}
\begin{tabular}{@{}lrrrrrr@{}}
\toprule
Policy / metric
& \shortstack{Audio /\\Video}
& \shortstack{Home\\Appliances}
& Laptops
& \shortstack{Mobile\\Phones}
& \shortstack{Tablets /\\E-readers}
& Overall \\
\midrule
\emph{Trained}
& 0.2363 & 0.2446 & 0.1802 & 0.1936 & 0.1987 & 0.2107 \\
\emph{Reference}
& 0.2303 & 0.2329 & 0.1577 & 0.1898 & 0.1848 & 0.1991 \\
\midrule
$\boldsymbol{\Delta u}$
& \textbf{+0.0060}
& \textbf{+0.0117}
& \textbf{+0.0225}
& \textbf{+0.0038}
& \textbf{+0.0139}
& \textbf{+0.0116} \\
$\boldsymbol{\Delta u(\%)}$
& \textbf{+2.59\%}
& \textbf{+5.01\%}
& \textbf{+14.26\%}
& \textbf{+1.98\%}
& \textbf{+7.55\%}
& \textbf{+5.81\%} \\
\shortstack[l]{95\% interval\\for $\Delta u$}
& \shortstack[r]{$[0.0036,$\\$0.0083]$}
& \shortstack[r]{$[0.0091,$\\$0.0143]$}
& \shortstack[r]{$[0.0172,$\\$0.0285]$}
& \shortstack[r]{$[0.0019,$\\$0.0057]$}
& \shortstack[r]{$[0.0115,$\\$0.0163]$}
& \shortstack[r]{$[0.0097,$\\$0.0136]$} \\
\bottomrule
\end{tabular}
\end{table}

\begin{figure}[!htp]
\centering
\includegraphics[width=\linewidth]{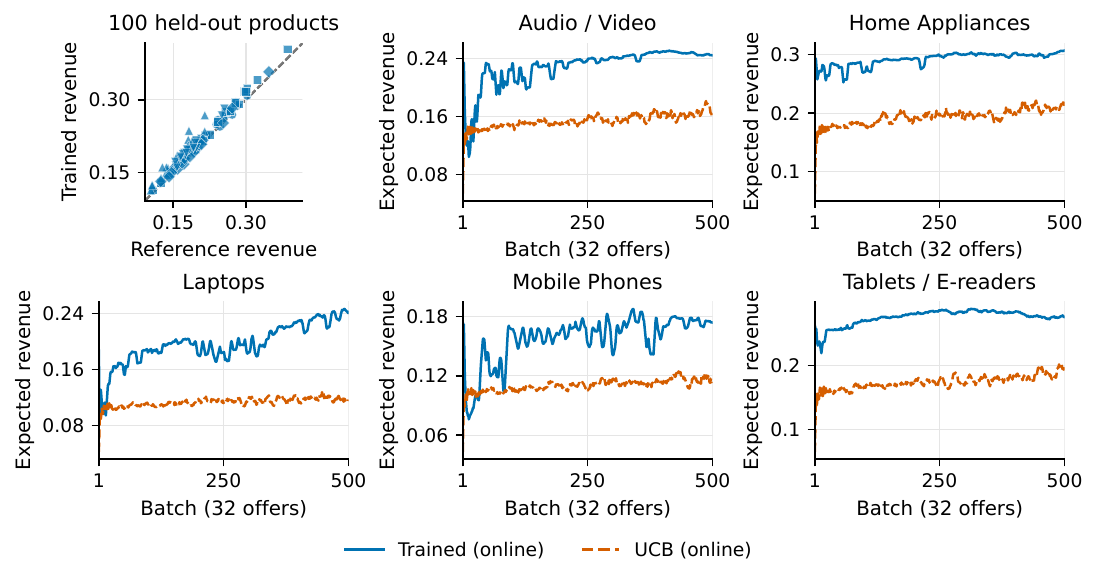}
\caption{Marketplace evaluation. Upper left: each point shows mean expected revenue for one test product (\texttt{Reference}, horizontal; \texttt{Trained}, vertical), with a dashed diagonal line indicating equal revenue. Markers follow the category order in Table~\ref{tab:exp-real}: circle, square, upward triangle, diamond, downward triangle. Other panels: training-revenue curves, which use the same colors as in Figure~\ref{fig:exp-synthetic}.}
\label{fig:exp-real}
\end{figure}
\FloatBarrier
\phantomsection
\label{exp:body-end}

\clearpage

\section*{AI use statement}
Generative AI tools were used during the drafting of this paper to improve the readability and refine the English exposition. We also used Codex to assist with code development. The authors bear full responsibility for the scientific content and the final manuscript.

\newpage
\appendix
\tableofcontents
\section{Unbiasedness of the Actor Gradient Estimator}
\label{app:unbiased-actor-gradient}

Fix a context $z$ and a batch $\tau$ in which the actor is updated.
Write $\theta(\psi):=\theta_{\mathrm{ref}}+\Delta\theta(\psi)$.
Let $\mathcal F_{\tau-1}$ denote the history before collecting
batch $\tau$, and abbreviate conditional expectation given this
history by $\mathbb E_\tau$. Conditional on $\mathcal F_{\tau-1}$,
the parameters $\psi_\tau,\phi_\tau$, the exploration probability
$\varepsilon_\tau$, and the exploration distribution $\nu_\tau$ are fixed.

Let $q_\tau^{\mathrm g}$ denote the greedy pricing function,
defined by
\[
q_\tau^{\mathrm g}(x,z)
\in \arg\max_{p\in[0,1]}
p\,\widehat d_{\phi_\tau}(p,x,z),
\]
using a fixed tie-breaking rule. For each $p\in[0,1]$,
let $q_p$ denote the constant pricing function $q_p(x,z)=p$.
At each interaction $b$, the seller chooses
\[
q_{\tau,b}
=
\begin{cases}
q_\tau^{\mathrm g},
& \text{with probability }1-\varepsilon_\tau,\\
q_{P_{\tau,b}},\quad P_{\tau,b}\sim\nu_\tau,
& \text{with probability }\varepsilon_\tau,
\end{cases}
\]
and sets $p_{\tau,b}=q_{\tau,b}(x_{\tau,b},z)$.
Conditional on $\mathcal F_{\tau-1}$, advertisements are sampled
independently from $\pi_{\theta_\tau}(\cdot\mid c(z))$.
The exploration decisions and the prices sampled from $\nu_\tau$
use fresh randomness, independent across interactions and
independent of advertisement generation. Purchase feedback satisfies
\[
\mathbb E_\tau[y_{\tau,b}\mid x_{\tau,b},p_{\tau,b}]
=
d(p_{\tau,b},x_{\tau,b},z).
\]

For a fixed advertisement $x$, the expected revenue under this
randomized choice of pricing function is
\begin{equation}
\label{eq:batch-expected-revenue}
\begin{aligned}
R_\tau(x,z)
&:=
\mathbb E_\tau\!\left[
q_{\tau,b}(x,z)d(q_{\tau,b}(x,z),x,z)
\right]\\
&=
(1-\varepsilon_\tau)
q_\tau^{\mathrm g}(x,z)
d(q_\tau^{\mathrm g}(x,z),x,z)+
\varepsilon_\tau
\int_{[0,1]}p\,d(p,x,z)\,\nu_\tau(dp).
\end{aligned}
\end{equation}
Accordingly, the KL-regularized objective in Equation~\eqref{eq:regularized-objective},
averaged over the batch's pricing randomization, is
\begin{equation}
\label{eq:batch-regularized-objective}
\begin{aligned}
J_\tau(\psi\mid z)
&:=
\mathbb E_\tau\!\left[
u_\beta(q_{\tau,b},\theta(\psi)\mid z)
\right]\\
&=
\mathbb E_{x\sim\pi_{\theta(\psi)}(\cdot\mid c(z))}
[R_\tau(x,z)] -
\beta D_{\mathrm{KL}}\!\left(
\pi_{\theta(\psi)}(\cdot\mid c(z))
\,\middle\|\,
\pi_{\theta_{\mathrm{ref}}}(\cdot\mid c(z))
\right).
\end{aligned}
\end{equation}
The pricing strategy, namely $q_\tau^{\mathrm g}$,
$\varepsilon_\tau$, and $\nu_\tau$, is held fixed when differentiating
with respect to $\psi$.

For the generation procedure of this paper, the policy $\pi_{\theta(\psi)}(\cdot\mid c(z))$ is differentiable in $\psi$ and has parameter-independent support equal to that of the reference policy. Indeed, the low-rank adaptation map $\Delta\theta$ is polynomial in $\psi$, and the actor’s logits and softmax probabilities are differentiable in their respective inputs. Thus, the probability of each fixed token sequence is differentiable in $\psi$. Moreover, generation uses the full vocabulary at the fixed temperature $0.7$, so every token has strictly positive probability for finite logits. Since the actor and the reference policy use the same vocabulary and stopping rules, both policies assign positive probability to exactly the same admissible sequences. Finally, the vocabulary is finite and generation stops after at most $64$ tokens. Hence, for each fixed context $z$, the expectations over generated sequences and the KL divergence reduce to finite sums of differentiable terms, with finite log ratios. Therefore, differentiation and expectation can be interchanged in the calculations below.

\begin{proposition}
For $B\geq 2$, the actor surrogate loss in Algorithm~\ref{alg:actor-critic} satisfies
\begin{equation}
\label{eq:actor-gradient-unbiasedness}
\mathbb E_\tau\!\left[
-\left.\nabla_\psi\mathcal L_{\mathrm{actor}}(\psi)
\right|_{\psi=\psi_\tau}
\right]
=
\left.\nabla_\psi J_\tau(\psi\mid z)
\right|_{\psi=\psi_\tau}.
\end{equation}
Here, the advantage estimates and log-ratio coefficients in the
surrogate loss are held fixed during differentiation.
\end{proposition}

\begin{proof}
Define the score and log-ratio functions
\[
s_\tau(x)
:=
\left.\nabla_\psi
\log\pi_{\theta(\psi)}(x\mid c(z))
\right|_{\psi=\psi_\tau},
\qquad
\ell_\tau(x)
:=
\log\frac{\pi_{\theta_\tau}(x\mid c(z))}
{\pi_{\theta_{\mathrm{ref}}}(x\mid c(z))}.
\]
The negative gradient of the surrogate loss is
\begin{equation}
\label{eq:actor-gradient-estimator}
\widehat g_\tau
=
\frac{1}{B}\sum_{b=1}^B
\left(A_{\tau,b}-\beta\ell_\tau(x_{\tau,b})\right)
s_\tau(x_{\tau,b}).
\end{equation}
Since policy probabilities sum to one, the score identity gives
\[
\mathbb E_\tau[s_\tau(x_{\tau,b})]
=
\left.\nabla_\psi
\sum_{x\in\mathcal X}\pi_{\theta(\psi)}(x\mid c(z))
\right|_{\psi=\psi_\tau}
=0.
\]

For $j\neq b$, the critic revenue estimate
\[
\widehat u_{\phi_\tau,j}
=
p_{\tau,j}\widehat d_{\phi_\tau}
(p_{\tau,j},x_{\tau,j},z)
\]
is conditionally independent of $x_{\tau,b}$, because the critic
is fixed before collecting the batch and advertisement draws
and pricing randomizations are independent across interactions.
Therefore,
\[
\mathbb E_\tau[
\widehat u_{\phi_\tau,j}s_\tau(x_{\tau,b})]
=
\mathbb E_\tau[\widehat u_{\phi_\tau,j}]
\mathbb E_\tau[s_\tau(x_{\tau,b})]
=0.
\]
Using
\[
A_{\tau,b}
=
r_{\tau,b}
-
\frac{1}{B-1}\sum_{j\neq b}\widehat u_{\phi_\tau,j},
\]
we conclude that the leave-one-out baseline contributes zero
to the conditional expected gradient. Formally,
\[
\begin{aligned}
\mathbb E_\tau[A_{\tau,b}s_\tau(x_{\tau,b})]
&=
\mathbb E_\tau[r_{\tau,b}s_\tau(x_{\tau,b})]
-
\frac{1}{B-1}\sum_{j\neq b}
\underbrace{
\mathbb E_\tau[
\widehat u_{\phi_\tau,j}s_\tau(x_{\tau,b})]
}_{=\,0}\\
&=
\mathbb E_\tau[r_{\tau,b}s_\tau(x_{\tau,b})].
\end{aligned}
\]
Moreover, the demand model and the pricing strategy imply
\(
\mathbb E_\tau[r_{\tau,b}\mid x_{\tau,b}]
=
R_\tau(x_{\tau,b},z).
\)
Consequently, the likelihood-ratio identity yields
\begin{equation}
\label{eq:revenue-gradient-identity}
\begin{aligned}
\mathbb E_\tau[A_{\tau,b}s_\tau(x_{\tau,b})]
&=
\mathbb E_\tau[r_{\tau,b}s_\tau(x_{\tau,b})]\\
&=
\mathbb E_\tau[
R_\tau(x_{\tau,b},z)s_\tau(x_{\tau,b})]\\
&=
\left.\nabla_\psi
\mathbb E_{x\sim\pi_{\theta(\psi)}(\cdot\mid c(z))}
[R_\tau(x,z)]
\right|_{\psi=\psi_\tau}.
\end{aligned}
\end{equation}

Finally, applying the product rule and using that the reference
policy is fixed gives
\begin{equation}
\label{eq:kl-gradient-identity}
\begin{aligned}
&\left.\nabla_\psi D_{\mathrm{KL}}\!\left(
\pi_{\theta(\psi)}(\cdot\mid c(z))
\,\middle\|\,
\pi_{\theta_{\mathrm{ref}}}(\cdot\mid c(z))
\right)\right|_{\psi=\psi_\tau}\\
&\quad=
\sum_x(\ell_\tau(x)+1)
\left.\nabla_\psi\pi_{\theta(\psi)}(x\mid c(z))
\right|_{\psi=\psi_\tau}\\
&\quad=
\mathbb E_\tau[
\ell_\tau(x_{\tau,b})s_\tau(x_{\tau,b})]
+
\underbrace{\mathbb E_\tau[s_\tau(x_{\tau,b})]}_{=\,0}\\
&\quad=
\mathbb E_\tau[
\ell_\tau(x_{\tau,b})s_\tau(x_{\tau,b})].
\end{aligned}
\end{equation}
The sum is over the policy support, and the last equality
uses the zero-mean score identity.
Taking conditional expectations in
Equation~\eqref{eq:actor-gradient-estimator} and applying
Equations~\eqref{eq:revenue-gradient-identity}
and~\eqref{eq:kl-gradient-identity} proves
Equation~\eqref{eq:actor-gradient-unbiasedness}, concluding the proof.
\end{proof}
The result requires neither an accurate nor an unbiased critic. Indeed,
the leave-one-out baseline has zero expected contribution for
any critic fixed before collecting the batch. Critic quality can
nevertheless affect the selected prices and the estimator's
variance. Unbiasedness holds at the sampling parameters
$\psi_\tau$, with the batch's pricing strategy held fixed.
\providecommand{\experimentroot}{experimental_section}
\raggedbottom
\section{Experimental Protocols and Additional Results}
\label{app:experiments}
We provide implementation and evaluation details for Section~\ref{sec:experiments}, followed by additional test summaries and the offline-data and baseline ablations. We retain the notation of Algorithm~\ref{alg:actor-critic}: $T$ counts individual offers, $N=T/B$ counts batches, $p$ is the normalized price, and $q$ denotes a pricing rule.

\subsection{Shared Implementation and Evaluation}
\label{app:shared}
We adapt \texttt{Qwen2.5-1.5B-Instruct} through rank-eight LoRA on the attention query and value projections, keeping the pretrained weights fixed. Disabling the adapters recovers \texttt{Reference}. Generation and sequence scoring use the same full-vocabulary distribution at temperature $0.7$, with log probabilities summed over generated tokens. In the actor surrogate, the revenue advantage and sampled KL log ratio are held fixed when differentiating the sequence log probability. Each actor update uses the current batch of $B=32$ offers and occurs every $K=5$ batches; each critic update uses $S=64$ offers sampled uniformly from the accumulated replay buffer, with replacement only while fewer than 64 offers are available. Both updates use Adam, giving 100 actor updates and 500 critic updates over $N=500$ batches. Table~\ref{tab:exp-settings} summarizes the shared settings. The main experiments use $\varepsilon_\tau=0$; Appendix~\ref{app:exploration} varies this schedule as offline data are removed.

\begin{table}[!htb]
\centering\small
\caption{Shared implementation settings across synthetic and marketplace experiments. Appendix~\ref{app:exploration} specifies changes to initial critic support, exploration, and anchoring in the offline-data ablation.}
\label{tab:exp-settings}
\vspace{7pt}
\begin{tabular}{@{}ll@{}}
\toprule
Quantity & Value \\
\midrule
Actor & \texttt{Qwen2.5-1.5B-Instruct} \\
Adapted matrices & Query and value projections; $L=56$ across 28 layers \\
LoRA rank / scaling coefficient & $r_\ell=8$ / $\alpha_\ell=16$ for every $\ell$ \\
Effective weight update & $(\alpha_\ell/r_\ell)B_\ell A_\ell=2B_\ell A_\ell$ \\
Query adapter dimensions & $A_\ell\in\mathbb R^{8\times1536}$, $B_\ell\in\mathbb R^{1536\times8}$ \\
Value adapter dimensions & $A_\ell\in\mathbb R^{8\times1536}$, $B_\ell\in\mathbb R^{256\times8}$ \\
Adapter initialization & $A_\ell$: Kaiming uniform; $B_\ell=0$ \\
LoRA dropout / reference weights & $0$ / $\theta_{\mathrm{ref}}$ frozen \\
Generation temperature / maximum new tokens & $0.7$ / $64$ \\
Maximum prompt length & $512$ tokens \\
Online horizon & $N=500$, $B=32$, $T=16{,}000$ \\
Critic replay minibatch / actor update period & $S=64$ / $K=5$ batches \\
Actor / online critic learning rate & $10^{-4}$ / $3\times10^{-3}$ \\
Actor KL / critic anchoring coefficient & $\beta=10^{-3}$ / $\rho=10^{-4}$ \\
Price grid & $101$ equally spaced admissible prices \\
Demand networks & Two hidden layers of 128 ReLU units, sigmoid output \\
Offline optimizer / learning rate / minibatch & Adam / $10^{-3}$ / $64$ \\
Offline stopping rule & 5 epochs without improvement in validation BCE \\
Checkpoint schedule & Every 100 batches, final evaluation after batch 500 \\
\bottomrule
\end{tabular}
\end{table}

For each context $z$, we evaluate unregularized expected revenue by averaging price times purchase probability over $m$ independently generated advertisements:
\begin{equation}
u(q,\theta\mid z)\approx\frac1m\sum_{i=1}^{m}q(x_i,z)d(q(x_i,z),x_i,z),
\qquad x_i\sim\pi_\theta(\cdot\mid c(z)).
\label{eq:exp-revenue-estimator}
\end{equation}
Here, $d$ is the exact synthetic demand or the frozen marketplace simulator. This evaluation integrates out Bernoulli purchase noise while retaining variation in generated advertisements, and excludes the KL penalty used during optimization. Both \texttt{Trained} and \texttt{Reference} use the same final critic to construct $q$, so their comparison isolates text adaptation under shared learned pricing. We report $\Delta u$ and $\Delta u(\%)$ as defined in Section~\ref{sec:experiments}, together with mean prices and purchase probabilities where relevant. Expected revenue averages the products $p\,d(p,x,z)$, rather than multiplying these two marginal means.

Synthetic tests use 10,000 advertisements per policy and environment, with 95\% percentile intervals for $\Delta u$ obtained from 2,000 independent bootstrap resamples of the two policies. Marketplace tests use 100 advertisements per policy and product, averaged first within products and then equally across products. Their 95\% intervals use 5,000 bootstrap resamples of paired product-level revenue differences, over 20 products per category or 100 products overall. These intervals condition on the trained policies and fixed demand model; they exclude variation across training seeds and simulator fits. The synthetic scatters group each policy's test advertisements into 100 nonoverlapping blocks of 100, aligned by sample index for display. This alignment does not define paired observations or independent training repetitions.

We implement \texttt{UCB} with one UCB1 arm per admissible grid price~\citep{auer2002finite}. Let $N_j$ and $\overline r_j$ denote the number of observations and mean observed reward for price $p_j$, and let $n=\sum_jN_j$ count completed offers. After selecting each of the 101 arms once in ascending price order, the next arm is
\begin{equation}
j_{n+1}\in\arg\max_j\left\{\overline r_j+\sqrt{\frac{2\log n}{N_j}}\right\}.
\label{eq:exp-ucb}
\end{equation}
Ties select the lowest arm index. The reference generator produces advertisements in batches, while arm statistics update after every offer using only the observed reward $r=py$. Price selection ignores the current advertisement, so each arm has a stationary reward distribution in $[0,1]$ under a fixed context and generator. Expected revenues are used only for evaluation and never enter arm selection. Both \texttt{Trained} and \texttt{UCB} receive 16,000 online outcomes.

\subsection{Synthetic Models, Calibration, and Reconstruction}
\label{app:synthetic}
The \emph{logit} and \emph{probit} models in Equation~\ref{eq:exp-choice-models} use a fixed product context and a binary feature $h(x)$, equal to one when advertisement $x$ contains the standalone word \textit{durable}, irrespective of case. Both steel-bottle contexts explicitly permit this claim. The coefficient $a$ controls baseline purchase propensity, $\kappa>0$ controls price sensitivity, and $w$ weights the content feature. These response families are motivated by credit-offer applications~\citep{farrell2021heterogeneity,bertrand2010advertising}; their numerical coefficients and the text feature are illustrative choices for the present experiments. Setting $(a_L,\kappa_L,w_L)=(1,4,0.8)$ makes logit demand decrease from $\sigma(1)$ to $\sigma(-3)$ over $p\in[0,1]$ when $h=0$, while the durability feature increases the index by $0.8$. This gives a positive, nonsaturating content effect and an interior revenue-maximizing price. We calibrate probit to match the resulting logit probabilities at $(p,h)=(0,0),(1,0),(1/2,1)$:
\begin{equation}
\begin{aligned}
a_P&=\Phi^{-1}(\sigma(1)), &
\kappa_P&=a_P-\Phi^{-1}(\sigma(-3)),\\
w_P&=\Phi^{-1}(\sigma(-0.2))-a_P+\kappa_P/2.
\end{aligned}
\label{eq:exp-probit-match}
\end{equation}
The resulting coefficients are $a_P=0.6160176927$, $\kappa_P=2.2863595488$, and $w_P=0.4019201292$. Calibration is deterministic and uses neither policy outcomes nor parameter estimates from the cited credit studies. The exact probit specification can be found in the repository.

For \emph{price disclosure}, the context describes the fictional paperback novel \emph{Stateline}. The prompt permits the token \texttt{[DISCLOSE\_PRICE]}, which is replaced with the selected dollar price after pricing. The binary feature records token presence, irrespective of case; the oracle does not validate other numerical or semantic claims in the advertisement. We construct demand as
\begin{equation}
d_{\mathrm{disclosure}}(p,x)=(1-h(x))\ell_0(5.99p)+h(x)\ell_1(5.99p),
\qquad p\in[0.99/5.99,1],
\label{eq:exp-book}
\end{equation}
where $\ell_0$ and $\ell_1$ map dollar prices to purchase probabilities for Plain Ad ($h=0$) and Price Ad ($h=1$), respectively. Each curve linearly interpolates between the three observations in Table~\ref{tab:exp-book-cells}; the factor $5.99$ converts normalized actions back to dollar prices. The difference between the two curves varies with price, allowing disclosure to affect both purchase propensity and price sensitivity. Interpolation preserves probabilities in $[0,1]$ within the observed range and is a modeling assumption applied to the published experimental cells. The construction does not reproduce the structural consumer-search model of \citet{dube2026effects}. We use 101 equally spaced prices on $[0.99/5.99,1]$ for price disclosure and on $[0,1]$ for logit and probit; rewards remain $py$ in normalized units. The published cell means are treated as fixed, so the reported intervals do not propagate their sampling uncertainty.

\begin{table}[!htb]
\centering\small
\caption{Price disclosure purchase probabilities from Table~A5, Panel~A of \citet{dube2026effects}. Plain Ad and Price Ad demand are interpolated separately over the three observed dollar prices.}
\label{tab:exp-book-cells}
\vspace{7pt}
\begin{tabular}{@{}lrrr@{}}
\toprule
Condition & \$0.99 & \$2.99 & \$5.99 \\
\midrule
Plain Ad, $h=0$ & 0.106 & 0.024 & 0.011 \\
Price Ad, $h=1$ & 0.147 & 0.039 & 0.012 \\
\bottomrule
\end{tabular}
\end{table}

For critic pretraining, we sample 1,024 unique training advertisements and 256 unique validation advertisements from \texttt{Reference} under the context used online, keeping the two text sets disjoint. Each advertisement receives one uniformly sampled price within each of eight equal-width strata of the admissible interval, followed by independent Bernoulli outcomes from the corresponding demand model. This produces 8,192 training offers and 2,048 validation offers. The critic receives the normalized price and declared feature $h(x)$, without access to the demand coefficients or exact probabilities. The synthetic prompt supplies two style examples and requests a factual title followed by one sentence; no marketplace observations or fitted marketplace networks enter these runs.

\begin{table}[!htb]
\centering\small
\caption{Synthetic test means over 10,000 advertisements per policy and environment. Prices are normalized; the feature indicates \textit{durable} in logit/probit and the disclosure token in price disclosure.}
\label{tab:exp-synthetic-detail}
\vspace{7pt}
\begin{tabular}{@{}llrrrr@{}}
\toprule
Demand & Policy & Revenue & Price & Purchase prob. & Feature rate \\
\midrule
Logit & \texttt{Trained} & 0.217117 & 0.555954 & 0.390138 & 0.971100 \\
Logit & \texttt{Reference} & 0.205418 & 0.534996 & 0.381832 & 0.821400 \\
Probit & \texttt{Trained} & 0.205616 & 0.647711 & 0.317174 & 0.967300 \\
Probit & \texttt{Reference} & 0.195496 & 0.637498 & 0.305357 & 0.821400 \\
Price disclosure & \texttt{Trained} & 0.027162 & 0.417170 & 0.065257 & 0.982300 \\
Price disclosure & \texttt{Reference} & 0.017416 & 0.469808 & 0.038771 & 0.351700 \\
\bottomrule
\end{tabular}
\end{table}

Table~\ref{tab:exp-synthetic-detail} shows how the revenue gains in Table~\ref{tab:exp-synthetic} combine price and content changes. Under logit and probit, both the mean selected price and purchase probability increase; under price disclosure, mean price decreases while purchase probability and disclosure frequency increase. In the latter model, the shared critic selects the same price for both policies conditional on $h$: approximately \$2.99 when $h=0$ and \$2.49 when $h=1$, yielding normalized revenues $u_0=0.01197997$ and $u_1=0.02743573$, respectively. If $f_{\texttt{Trained}}$ and $f_{\texttt{Reference}}$ denote the respective test disclosure frequencies, then
\begin{equation}
\Delta u=(f_{\texttt{Trained}}-f_{\texttt{Reference}})(u_1-u_0)
=(0.9823-0.3517)(u_1-u_0)\simeq0.00974640.
\label{eq:exp-disclosure-gain}
\end{equation}
Thus, the \textbf{55.96\%} revenue gain is fully explained by the change in feature frequency under shared learned pricing. This identifies adaptation to the declared demand mechanism, without measuring general advertisement quality.

\phantomsection
\label{app:ucb-reconstruction}
The synthetic \texttt{UCB} logs retain selected arms, prices, and observed rewards, but omit generated advertisements and their purchase probabilities. We therefore estimate marginal training revenue using the logged prices and the offline reference feature frequency. Let $\widehat\alpha$ be the fraction of the 1,024 offline training advertisements with $h(x)=1$. For a recorded price $p_n$, the reconstructed expected revenue is
\begin{equation}
\widetilde u_n=p_n\big[(1-\widehat\alpha)d(p_n,h=0)+\widehat\alpha d(p_n,h=1)\big],
\label{eq:exp-ucb-reconstruction}
\end{equation}
where $d(p,h)$ denotes the corresponding synthetic demand at feature value $h$. The offline feature frequencies are $0.8251953125$ for logit and probit and $0.3525390625$ for price disclosure. Marginalizing over content is appropriate because \texttt{UCB} selects prices independently of the current advertisement. The offline corpus is deduplicated, so its feature frequency is an approximation to that of unrestricted online generations. We average successive groups of 32 offers and apply the same trailing ten-batch window used for \texttt{Trained}. No test revenues or test feature frequencies enter this reconstruction; the \texttt{Trained} curves use recorded offer probabilities directly.

\subsection{Marketplace Data and Evaluation Protocol}
\label{app:real-protocol}
We use Avito marketplace listings~\citep{avito2018demand} containing a title, description, posted price, and supplied deal-probability target. Preprocessing retains the five categories in Section~\ref{sec:exp-real} and removes missing required fields, nonfinite or nonpositive prices, targets outside $[0,1]$, duplicate listing identifiers, and source advertisements outside 8-60 words. Reviewed annotations, fixed before the experiments, group listings by specific commercial product identity using title, category, and product parameters, without access to price or demand targets. Ambiguous identities are omitted; conditions and minor variants may share a product group. We trim prices outside the first and ninety-ninth percentiles of the selected listings and retain at most one listing per seller and product. Complete titles and descriptions are translated into English with OPUS-MT Russian-to-English; the build stops if Cyrillic characters remain. The resulting dataset contains 20,826 listings and 393 products. Images, timestamps, locations, and seller attributes are excluded from model inputs.

Each context $z$ contains a canonical product name, category, and a label describing its physical condition, annotated from the retained listing with the smallest identifier. The condition describes this representative context rather than every listing in the product group. The current product context contains no source price, probability target or source advertisement. Five fixed marketplace examples, one per category, guide generation format and phrasing; their entire product groups are excluded before supervised and actor splits. We denote this set by $\mathcal S$, with $|\mathcal S|=5$, and include it in the prompt $c(z,\mathcal S)$ shown in Figure~\ref{fig:exp-prompt}. The prompt requires product facts to come exclusively from $z$ and requests a title followed by a one-sentence description. Parsed outputs are canonicalized as title followed by description before embedding, matching the offline advertisements. The examples provide a common generation format and vocabulary.

\begin{figure}[h]
\centering
\begingroup\setlength{\fboxsep}{8pt}\setlength{\fboxrule}{0.4pt}
\fbox{\begin{minipage}{\dimexpr\linewidth-2\fboxsep-2\fboxrule\relax}
\small\ttfamily\raggedright\setlength{\parskip}{2pt}\setlength{\parindent}{0pt}
Create one advertisement from the current context.\par
Rules:\par
- Use only information explicitly provided in the current context.\par
- Treat missing or unspecified values as unavailable. Do not infer or mention them.\par
- Use the examples only to learn the required syntax, structure, length, and writing style.\par
- Never copy product-specific facts from the examples.\par
- Do not add unsupported specifications, prices, discounts, warranties, accessories, contact details, or promotional claims.\par
- Follow the output format exactly.\par
- Return only the advertisement, without explanations or additional text.\par
Language:\par
English\par
Output format:\par
Title: \textless{}a short factual title naming the product\textgreater{}\par
Description: \textless{}one short factual sentence\textgreater{}\par
The complete advertisement must contain between 8 and 60 words.\par
Syntax demos:\par
$\langle\mathcal S\rangle$\par
Current context:\par
$\langle z\rangle$\par
Advertisement:
\end{minipage}}\endgroup
\caption{Marketplace generation prompt, with language and output-format fields expanded. The symbolic fields $\mathcal S$ and $z$ contain the five fixed marketplace examples and current product context, respectively.}
\label{fig:exp-prompt}
\end{figure}

Listings are randomly split within category after sorting identifiers, using seed 47 and proportions $35/10/35/10/10\%$ for environment training/validation, critic training/validation, and the common held-out set. These sets contain 7,291, 2,081, 7,291, 2,081, and 2,082 listings, respectively. Listing identities are disjoint, although different listings of the same product may occur in several sets. The demand simulator (environment) and initial critic are fitted independently against the supplied probability targets, selecting checkpoints by validation binary cross-entropy (BCE). The environment uses seed 48 and selects epoch 4; the initial critic uses seed 49 and selects epoch 12. The environment is then frozen, and online feedback is sampled as $y\sim\operatorname{Ber}(d(p,x,z))$. Eligible actor products require at least five listings, two distinct prices, two distinct descriptions and at least one listing in each supervised training set. Selection uses seed 48 and no demand targets or policy outcomes, assigning one training product and 20 other test products per category. The training products are Apple iPod Touch~5, Chaika~132M, Asus~X54H, Sony Xperia~C4, and Oysters T72X~3G, in the category order of Table~\ref{tab:exp-real}. Each run starts from the same reference actor and initial critic, with a seed determined by base seed 47 and context identity. Evaluation freezes the category's final actor and critic and uses each test product's own context and price scale.

Both demand networks receive the monetary price and a frozen advertisement embedding, obtained by averaging the hidden states from the last layer of the frozen reference Qwen model over nonpadding tokens. LoRA adapters are disabled during encoding, which is limited to 384 tokens. Writing this 1,536-dimensional embedding as $h_{\texttt{Qwen2.5}}(x)$, the critic takes the form
\begin{equation}
\widehat d_\phi(p,x,z)=\sigma\!\left(\operatorname{MLP}_\phi([h_{\texttt{Qwen2.5}}(x),\lambda(z)p])\right),
\label{eq:exp-real-network}
\end{equation}
with two hidden layers of 128 ReLU units; the environment has the same architecture and independently fitted parameters. Neither network receives a separate context vector, so context enters through the generated text and product-specific price conversion. The scale $\lambda(z)$ is the maximum retained listing price for the product, computed across all listing splits without using demand targets and held fixed throughout training and evaluation. Thus, the protocol assumes an available product price range, including for products unseen during actor optimization. Dividing monetary revenue $\lambda(z)p\,d(p,x,z)$ by $\lambda(z)>0$ preserves the unregularized revenue-maximizing price within $[0,\lambda(z)]$. Aggregate normalized revenue weights products equally relative to their price scales; the fixed KL coefficient also applies in these normalized units. The networks impose no price-monotonicity constraint, and prices below a product's historical minimum remain extrapolations.

Table~\ref{tab:exp-real-intervals} gives the bootstrap intervals over products at greater precision and the number of products with positive gains. \texttt{Trained} improves on 91 of 100 products: all 20 products in Home Appliances, Laptops, and Tablets/E-readers, 16 in Audio/Video, and 15 in Mobile Phones. Figure~\ref{fig:exp-real-detail} shows the corresponding product means. Thus, the positive overall gain is accompanied by positive category means, while the nine individual losses occur in Audio/Video and Mobile Phones.

\begin{figure}[!htbp]
\centering
\includegraphics[width=.72\linewidth]{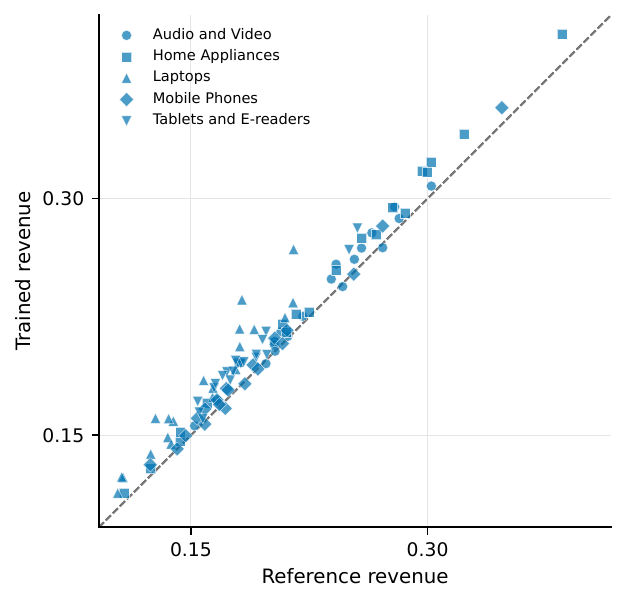}
\caption{Marketplace held-out revenues for the 100 test products (\texttt{Reference}, horizontal; \texttt{Trained}, vertical), with category-specific markers. Each coordinate averages 100 advertisements per policy and product; the dashed diagonal line indicates equal expected revenue.}
\label{fig:exp-real-detail}
\end{figure}

\begin{table}[!htb]
\centering\small
\caption{Marketplace revenue gains and improved-product counts for \texttt{Trained} over \texttt{Reference}. The 95\% intervals use 5,000 paired-product bootstrap resamples.}
\label{tab:exp-real-intervals}
\vspace{7pt}
\begin{tabular}{@{}lrrr@{}}
\toprule
Category & $\Delta u$ & 95\% interval for $\Delta u$ & Improved products \\
\midrule
 Audio/Video & \textbf{+0.005956} & [0.003590, 0.008313] & 16/20 \\
Home Appliances & \textbf{+0.011678} & [0.009055, 0.014272] & 20/20 \\
Laptops & \textbf{+0.022490} & [0.017165, 0.028486] & 20/20 \\
Mobile Phones & \textbf{+0.003765} & [0.001890, 0.005677] & 15/20 \\
 Tablets/E-readers & \textbf{+0.013943} & [0.011474, 0.016290] & 20/20 \\
Overall & \textbf{+0.011566} & [0.009672, 0.013632] & 91/100 \\
\bottomrule
\end{tabular}
\end{table}

\FloatBarrier
\subsection{Progressive Scarcity of Offline Data}
\label{app:exploration}
We assess sensitivity to the amount of data available for critic pretraining using Laptops, the category with the median number of critic-training listings. The actor-training product remains Asus~X54H, evaluation uses the same 20 held-out Laptop products, and the environment retains its full-data fit. From the $M=7{,}291$ critic-training listings across all categories, we construct subsets of size $D$, stratified by category and selected without using demand targets, corresponding approximately to 100, 75, 50, 25, and 0\% of the original data. The missing-data fraction is $g(D)=1-D/M$. We set the exploration distribution $\nu_\tau$ to be uniform on the 101-price grid and use
\begin{equation}
\varepsilon_\tau(D)=\frac{g(D)}{\sqrt{\max\{1,\tau-1\}}},
\qquad \tau\in[N].
\label{eq:exp-epsilon}
\end{equation}
With the remaining probability $1-\varepsilon_\tau(D)$, the critic selects the price maximizing estimated revenue for the current advertisement. The schedule increases exploration as offline support decreases and reduces it as online feedback accumulates. Its expected number of uniform-price offers is approximately $2Bg(D)\sqrt N$; this is a description of the schedule's scale, not a theoretical performance guarantee. At zero data, the critic starts from seeded random parameters and uses $\rho=0$, avoiding anchoring to an unfitted model. All runs receive 16,000 online outcomes and the same actor and critic update counts; the full-data endpoint has $\varepsilon_\tau=0$ and reuses the completed main run. Evaluation freezes the final actor and critic and disables price exploration.

Figure~\ref{fig:exp-exploration} and Table~\ref{tab:exp-exploration} report expected revenue over the same 20 products, using 100 advertisements per policy and product. Revenue ranges from $0.1802$ to $0.1895$ with 25-100\% of the initial data available, and falls to $0.1670$ at zero data, \textbf{7.32\%} below the full-data control. The algorithm therefore continues learning from online outcomes without supervised initialization, with a lower observed endpoint. The intermediate results are not monotone, and the single run per budget with overlapping product-bootstrap intervals does not establish that removing offline data improves performance. Because initial support and the exploration schedule vary together, the comparison measures the complete learning procedure under missing data rather than the isolated effect of exploration.

\begin{figure}[!htbp]
\centering
\includegraphics[width=.9\linewidth]{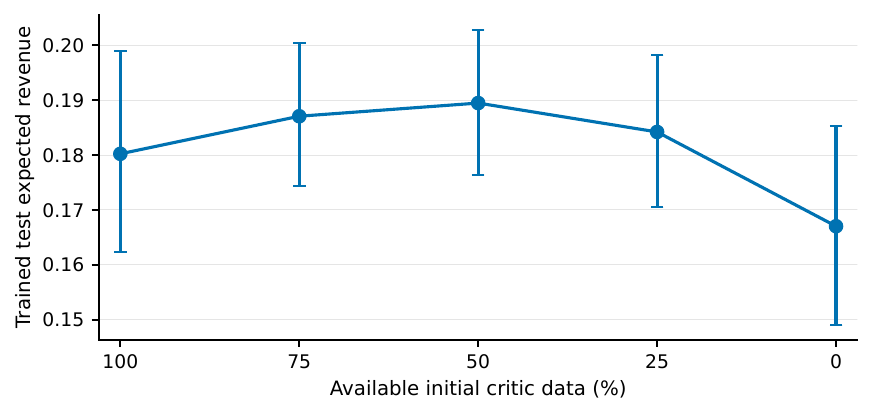}
\caption{\texttt{Trained} expected revenue as initial critic data decrease from 100\% to 0\%, using the exploration schedule in Equation~\ref{eq:exp-epsilon}. Points average the same 20 held-out Laptop products; the error bars show 95\% product-bootstrap intervals conditional on one training run per data budget.}
\label{fig:exp-exploration}
\end{figure}

\begin{table}[!htb]
\centering\small
\caption{Initial critic training data, offers sampled from the uniform exploration distribution and \texttt{Trained} expected revenue with 16,000 online outcomes. The 95\% intervals use 5,000 bootstrap resamples of the 20 held-out Laptop products.}
\label{tab:exp-exploration}
\vspace{7pt}
\begin{tabular}{@{}rrrrrr@{}}
\toprule
Initial data & All listings & Laptop listings & Uniform offers & Revenue & 95\% interval \\
\midrule
100\% & 7,291 & 190 & 0 & 0.180224 & [0.1624, 0.1990] \\
75\% & 5,468 & 142 & 342 & 0.187076 & [0.1743, 0.2004] \\
50\% & 3,645 & 95 & 703 & 0.189475 & [0.1764, 0.2028] \\
25\% & 1,823 & 48 & 1,068 & 0.184210 & [0.1705, 0.1982] \\
0\% & 0 & 0 & 1,409 & 0.167030 & [0.1490, 0.1854] \\
\bottomrule
\end{tabular}
\end{table}

\FloatBarrier
\subsection{Critic Predictions in the Leave-One-Out Baseline}
\label{app:baseline}
We examine the choice in Equation~\ref{eq:jpa-revenue-advantage} to average the critic's predicted revenues in the leave-one-out baseline. The diagnostic compares critic estimation error with the Bernoulli noise of an observed-reward baseline, keeping advertisements and prices fixed and performing no actor retraining. For each category, we sample 1,000 seeded batches of $B=32$ historical offers from the common held-out listing set, with replacement and reuse these batches across the five initial critic-data budgets in Appendix~\ref{app:exploration}. Each sampled occurrence is treated as an independent Bernoulli offer, including repeated listings; outcome noise is integrated analytically. Let $d_j=d(p_j,x_j,z_j)$ be the frozen environment probability and $\widehat d_j=\widehat d_{\phi_1}(p_j,x_j,z_j)$ the purchase probability predicted by the initial critic at the selected data budget. For a fixed batch and position $b\in[B]$, the target baseline and its two estimators are
\begin{align}
L_b&=\frac1{B-1}\sum_{j\ne b}p_jd_j, &
L_b^{\mathrm{obs}}&=\frac1{B-1}\sum_{j\ne b}p_jy_j, &
L_b^{\mathrm{crit}}&=\frac1{B-1}\sum_{j\ne b}p_j\widehat d_j.
\label{eq:exp-baseline-estimators}
\end{align}
Conditional on the sampled offers, their mean squared errors relative to $L_b$ are
\begin{equation}
\begin{aligned}
\operatorname{MSE}(L_b^{\mathrm{obs}})&=\frac{\sum_{j\ne b}p_j^2d_j(1-d_j)}{(B-1)^2},\quad \operatorname{MSE}(L_b^{\mathrm{crit}})=\bigg(\frac1{B-1}\sum_{j\ne b}p_j(\widehat d_j-d_j)\bigg)^2.
\end{aligned}
\label{eq:exp-baseline-mse}
\end{equation}
The first expression is the variance due to Bernoulli purchase outcomes and the second is squared critic error. We average each over batches and all $B$ positions, then divide the mean critic error by the mean error of the baseline using observed rewards. Values below one indicate that critic error is smaller than the Bernoulli noise obtained by averaging rewards, as in standard REINFORCE leave-one-out (RLOO).

\begin{figure}[H]
\centering
\includegraphics[width=.9\linewidth]{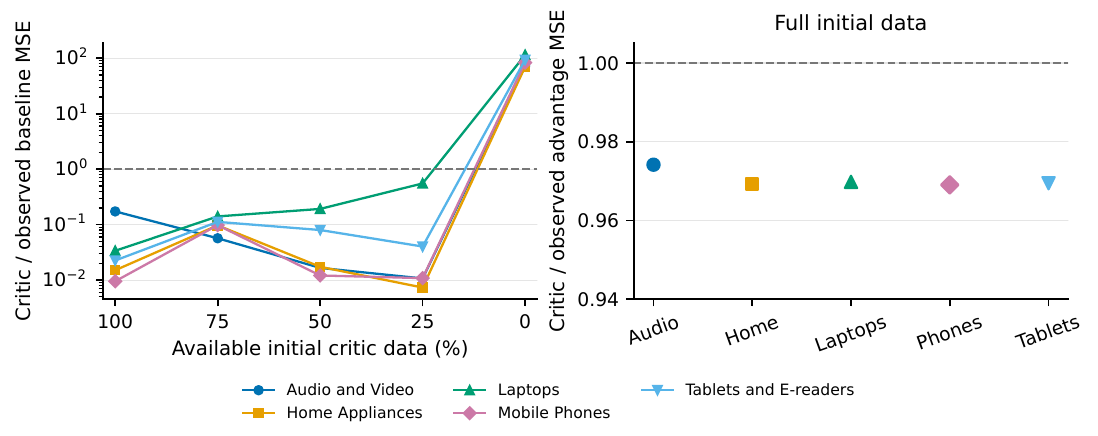}
\caption{Baseline and revenue-advantage MSE ratios for critic-based versus observed-reward estimators. Left: baseline MSE as offline support decreases, on a logarithmic scale. Right: advantage MSE at full data, including purchase noise from the offer whose advantage is estimated. The dashed line indicates equal conditional MSE.}
\label{fig:exp-baseline}
\end{figure}

At full data, the baseline MSE ratios are $0.1734$, $0.0151$, $0.0337$, $0.0095$, and $0.0226$ in the category order of Table~\ref{tab:exp-real}. All categories remain below one with 25-75\% data, although the ratios are not monotone. At zero data, the randomly initialized critic has approximately 70-115 times the observed-baseline error. Figure~\ref{fig:exp-baseline} therefore shows a substantial reduction in baseline error after supervised fitting, together with the failure of the same substitution for an unfitted critic. Purchase noise from offer $b$ remains under either baseline: relative to the target $p_bd_b-L_b$,
\begin{equation}
\operatorname{MSE}\!\left(p_by_b-L_b^{\cdot}\right)
=p_b^2d_b(1-d_b)+\operatorname{MSE}(L_b^{\cdot}),
\qquad \cdot\in\{\mathrm{obs},\mathrm{crit}\}.
\label{eq:exp-advantage-mse}
\end{equation}
For Laptops at full data, baseline MSE decreases from $0.002256$ to $0.000076$, while revenue-advantage MSE decreases from $0.072191$ to $0.070011$, a reduction of \textbf{3.02\%}. This common noise floor limits the effect of baseline improvement on the complete advantage. The diagnostic excludes the KL term and the gradient of the log policy probability, and its historical batches may mix product contexts within a category. It therefore measures baseline and advantage error without establishing a reduction in full policy-gradient variance or an increase in final-policy revenue.

\FloatBarrier
\phantomsection
\label{exp:appendix-end}
\end{document}